\documentclass{sig-alternate-05-2015}

\usepackage{graphicx}
\usepackage{algorithm}
\usepackage{algorithmic}
\usepackage{color}
\usepackage{times}
\usepackage{amsmath}
\usepackage{enumitem}

\begin{document}



\title{Effective Keyword Search in Graphs}

\author{
\alignauthor
Mehdi Kargar{\small $^{*}$}, Lukasz Golab{\small $^{\#}$}, Jaroslaw Szlichta{\small $^{\$}$}\\
       \affaddr{$^{*}$York University, Toronto, Canada}\\
       \affaddr{$^{\#}$University of Waterloo, Waterloo, Canada}\\
       \affaddr{$^{\$}$University of Ontario Institute of Technology, Oshawa, Canada}\\              
       \affaddr{kargar@cse.yorku.ca, lgolab@uwaterloo.ca, jaroslaw.szlichta@uoit.ca}
}

\maketitle

\newdef{definition}{Definition}
\newdef{problem}{Problem}
\newtheorem{theorem}{Theorem}

\begin{abstract}

In a node-labeled graph, keyword search finds subtrees of the graph whose nodes contain all of the query keywords.  This provides a way to query graph databases that neither requires mastery of a query language such as SPARQL, nor a deep knowledge of the database schema.  Previous work ranks answer trees using combinations of structural and content-based metrics, such as path lengths between keywords or relevance of the labels in the answer tree to the query keywords.   We propose two new ways to rank keyword search results over graphs.  The first takes node importance into account while the second is a bi-objective optimization of edge weights and node importance.  Since both of these problems are NP-hard, we propose greedy algorithms to solve them, and experimentally verify their effectiveness and efficiency on a real dataset.

\end{abstract}



\section{Introduction}

Many interesting datasets, e.g., Web and social data, describe connections and relationships, and therefore can be naturally represented as graphs.  Graph databases have recently been proposed to natively manage such data, with SPARQL being a popular query language.  However, users who are not familiar with the query language or the database schema risk being locked out of valuable data.  This is where keyword search comes in: as an intuitive way to query and explore graphs whose nodes are associated with text, without having deep knowledge of programming languages or the database schema.  While the keyword search approach certainly has merit, a technical challenge is to extract \emph{semantically meaningful} subsets of the data that contain the specified keywords.  Given the limitations of existing tools,  we present a novel framework that allows non-technical users to effectively search for answers in graph databases.

Given a set of query keywords, the problem is to find a subtree of the graph that contains all the keywords.  An answer tree consists of a \emph{root node}, \emph{content nodes} at the leaves, which contain the query keywords, and \emph{middle nodes} that connect the content nodes together through the root node.  By convention, the root node is also considered to be a middle node.

Previous work ranks answer trees using combinations of structural and content-based metrics, such as the sum of edge weights (shorter weighted paths from the root to the content nodes are better), and IR metrics corresponding to the relevance of the text in the answer tree to the query keywords \cite{He2007}.  However, different nodes in the graph may have different semantic importance to the answer.  Following on and inspired by recent work in keyword search in relational databases \cite{Kargar2015a}, we have implemented a system for effective keyword search over graph databases, that takes node importance, including the middle nodes, into account.

\begin{figure}[t]
\centering
\includegraphics[width=8.5cm]{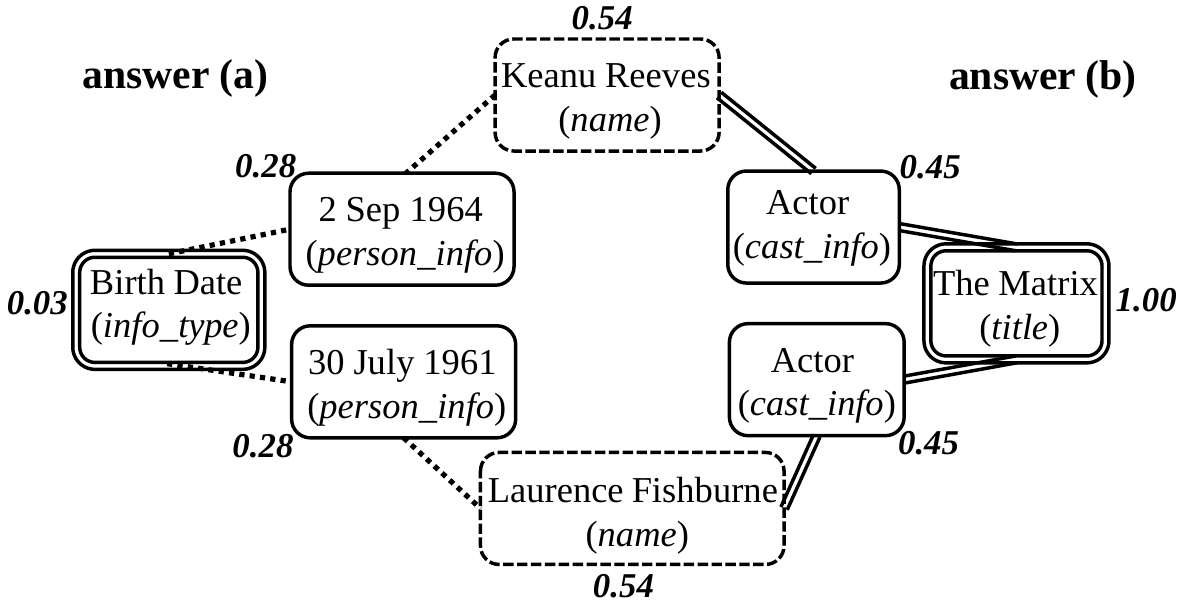}
\caption{Two answers of the example keyword query.} \label{fig:graph}
\end{figure}

To motivate our approach and illustrate the shortcomings of existing techniques, consider a graph obtained from the Internet Movie Database (IMDb) by converting \emph{foreign keys} in the database schema to undirected equally-weighted edges as shown in Figure \ref{fig:graph}.  For each node, the name of its table in the database is also shown (e.g., The Matrix movie
appeared in the \textit{title} table).  In total, IMDb has 21 tables associated with 21 different node types in the graph.  Suppose we issue the keyword query [\textit{"Keanu Reeves" "Laurence Fishburne"}] to find the relationship(s) between these two actors.  Two answers, labeled (a) (rooted at Birth Date and using dotted edges) and (b) (rooted at The Matrix and using double edges), are shown in Figure \ref{fig:graph}.  Answer (a) suggests that Reeves and Fishburne are connected because they both have an \textit{info\_type} called Birth Date in the \textit{person\_info} table, though their actual birth dates are different.  Answer (b) reveals that these two actors both starred in the movie The Matrix.  Clearly, answer (b) is more meaningful.  However, previous work cannot distinguish between these two answers because their content nodes have the same keywords (i.e., the same IR score) and they both have the same number of edges.  The difference is in the middle nodes: as shown in Figure~\ref{fig:graph}, The Matrix (\textit{title}) and Actor (\textit{cast\_info}) nodes have higher weights than \textit{info\_type} and \textit{person\_info} nodes because acting in the same movie is a strong relationship\footnote{We computed the node importance values in Figure~\ref{fig:graph} from the importance of the corresponding tables in the IMDb database, which in turn we computed using the algorithm from \cite{Kargar2015a}.  In general, node importance is application-specific; e.g., social credibility of people in a social graph or PageRank scores in the WWW graph.}.

Our contributions are as follows.

\begin{enumerate}[topsep=0pt,itemsep=-1ex,partopsep=1ex,parsep=1ex]

\item 
We design a system for effective keyword search in graphs.  The system provides two new ways to rank the results: the first one takes node importance into account whereas the second one is a bi-objective optimization of edge weights and node importance (Section 2). We prove that both of these problems are NP-hard (Section 2). Therefore, we propose two greedy solutions to rank answers (Section 3).

\item
We ensure the efficiency of our ranking algorithms by implementing a 2-hop cover index that returns the shortest path distance between any two nodes.  Since answering distance queries in graphs is known to be computationally expensive \cite{Akiba2013}, we propose several optimizations that significantly improve performance and scalability without compromising the precision of answers. These include reducing the search space by indexing only the nodes that are close to each other; this is reasonable since answer trees with nodes far from each other are unlikely to be chosen.  Furthermore, our second greedy algorithm does not need to build separate indexes for different values of the parameter controlling the importance of edge weights and node importance (Section 3).

\item The system provides new relevance measures for the importance of nodes and edges based on our prior work on keyword search over databases \cite{Kargar2015a}.

\item 
We perform a comprehensive evaluation based on IMDb dataset to demonstrate the viability of our methods and to compare against existing methods.
  
\end{enumerate}

\section{Problem Statement}

Given a node-labeled graph with application-supplied node and edge weights, as well as a query consisting of a set of keywords, we want to find and rank a set of \textbf{\textit{trees}} that contain all the query keywords. 
Assuming that edge weights denote semantic distance between nodes, effective answer trees should have a small sum of edge weights and a large sum of node (importance) weights.

\begin{definition}
\textbf{Edge Weight Objective (EW):} 
Suppose that the edges of an answer tree $T$ are denoted as $\{e_1, e_2, \dots$ $, e_m\}$. The edge weight of the answer is defined as $EW(T) = \sum_{i=1}^{m} w(e_i)$, where $w(e_i)$ is the weight of edge $e_i$. 
\label{definition-edge-weight}
\end{definition}

Minimizing the edge weight objective is NP-hard \cite{He2007}. Now, assume that $imp(n_i)$ is the importance of node $n_i$ and let $imp'(n_i) = \frac{1}{imp(n_i)}$.

\begin{definition} 
\textbf{Node Importance Objective (NI):} 
Suppose that the nodes (content and middle) of an answer tree $T$ are denoted as $\{n_1, n_2, \dots , n_k\}$. The node importance of the answer tree is defined as $NI(T) = \sum_{i=1}^{k} imp'(n_i)$.
\label{definition-node-imp}
\end{definition}

\begin{problem}
Given a graph $G$ and a set of query keywords, find a minimal tree $T$ in $G$ that covers the keywords and minimizes the node importance objective $NI(T)$.
\label{problem-nodeimp}
\end{problem}

\begin{theorem}
Problem \ref{problem-nodeimp} is NP-hard.
\end{theorem}

We prove that the decision version of the problem is NP-hard. Thus, as a direct result, minimizing node importance objective is NP-hard too. The decision problem is specified as follows.

\begin{problem}
Given a graph $G$ and a set of input keywords, determine whether there exists a minimal tree with node importance value of $ni$, for some constant $ni$.
\label{problem-nodeimp-decision}
\end{problem}

\begin{theorem}
Problem \ref{problem-nodeimp-decision} is NP-hard.
\end{theorem}

\begin{proof}
The problem is obviously in NP. We prove the theorem by a reduction from group Steiner tree problem. First, consider a graph in which all edges have the same weight of 1.0 and all nodes have the same importance of 1.0. A feasible solution to the above problem with the node importance at most $ni$ is a solution for the group Steiner tree problem with the weight at most ($ni$ - 1). This is the case since for any tree, the number of edges is equal to the number of nodes minus one. Thus, if there exists a tree with the node importance at most $ni$, then there exists a tree with the sum of the edge weights at most ($ni$ - 1). On the other hand, a tree with edge weights at most ($ni$ - 1) determines a feasible tree with the node importance at most $ni$. Therefore, the proof is complete.
\end{proof}

Now, we are also interested in the bi-criteria optimization problem of maximizing node importance and minimizing the sum of edge weights. One way to solve a bi-criteria optimization problem is to convert it into a single objective problem by combining the two objective functions into one. We define a single objective that combines the edge weights and node importance with a tradeoff parameter $\lambda$ as follows.

\begin{definition} 
\textbf{Combined Objective (C):}
Given an answer tree $T$ from graph $G$ for a given set of input keywords and a tradeoff $\lambda$ between the edge weights and the node importance, the combined objective of $T$ is defined as $C(T) = \lambda.NI(T) + (1-\lambda).EW(T)$. \label{definition-combined-objective}
\end{definition}

The parameter $\lambda$ varies from 0 to 1 and determines the tradeoff between edge weights and the node importance.  Since $\lambda$ is application-dependent, we leverage user and domain expert feedback to set and update it over time.  Incorporating user feedback is a vital component towards achieving high search precision.  Since $EW(T)$ and $NI(T)$ may have different scales, they should first be normalized.  Given the combined cost objective, we define the following problem:

\begin{problem}
Given a graph $G$, a set of input keywords and a tradeoff $\lambda$ between the edge weight and node importance objectives, find a minimal tree $T$ in $G$ that covers the input keywords and minimizes the combined objective $C(T)$.
\label{problem-combined}
\end{problem}

\begin{theorem}
Problem \ref{problem-combined} is NP-hard.
\end{theorem}

\begin{proof}
We show that finding a tree covering the input keywords with minimized edge weight objective ($EW(T)$) or minimized node importance objective ($NI(T)$) is NP-hard. Since both $EW(T)$ and $NI(T)$ are linearly related to $C(T)$ (the objective of Problem \ref{problem-combined}), then minimizing $C(T)$ is also an NP-hard problem.
\end{proof}

\section{System Description}

\begin{figure*}[t]
\centering
\includegraphics[width=10cm]{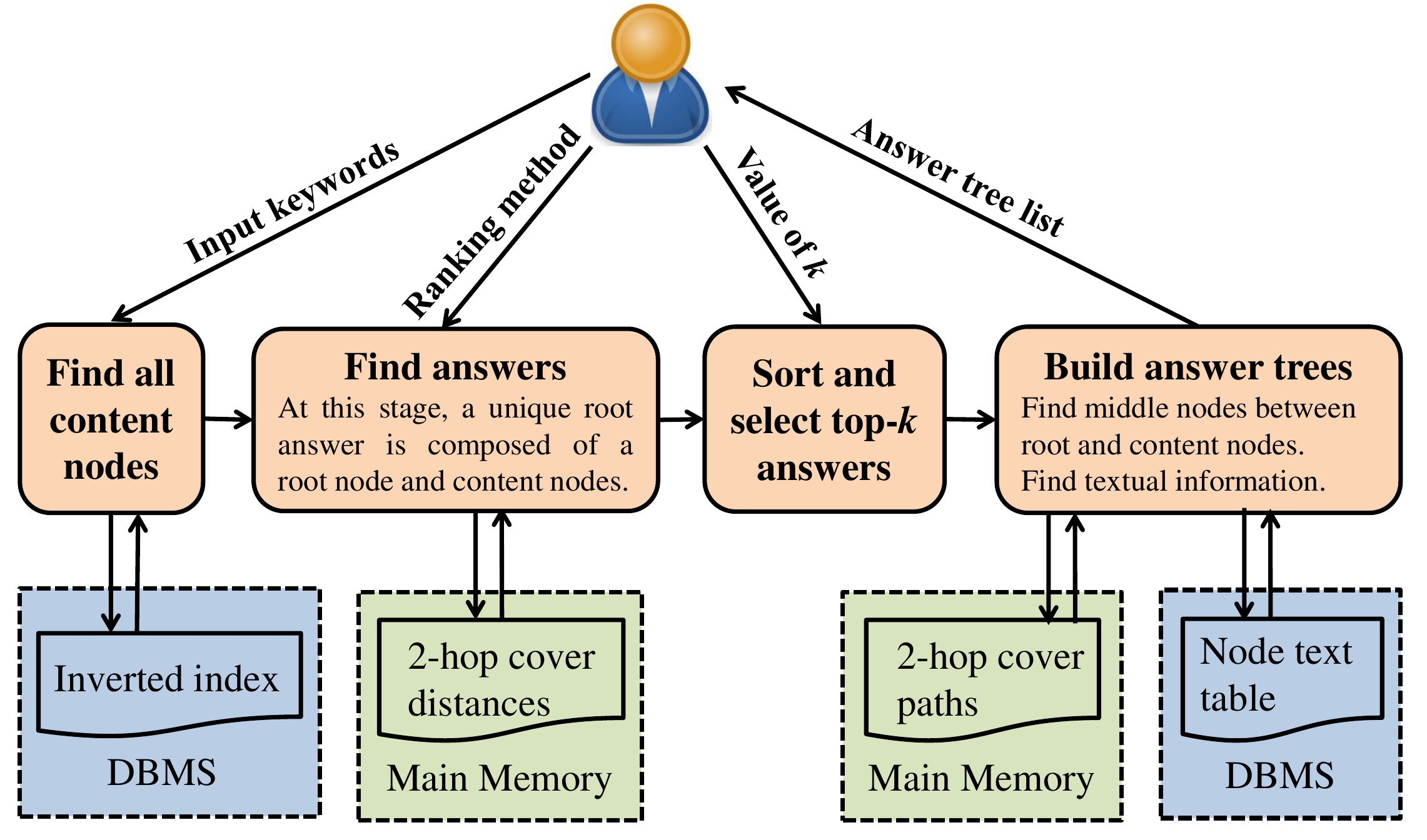}
\caption{System Architecture.} \label{fig:architecture}
\end{figure*}

Figure~\ref{fig:architecture} shows the system architecture of the system, consisting of four functional components and four data components.  For each query keyword, the system finds the IDs of the (content) nodes that contain the keyword using an inverted index.  These node IDs are stored in a map in which the key is the keyword and the value is a set of node IDs that contain the keyword.  This map is passed to the procedure for finding unique answer trees.  Using the index for finding the shortest distances between potential  root nodes and content nodes, the second functional component finds a set of unique root answers.  At this stage, each answer is composed of a root node and a set of content nodes.  This set is passed to the next step to sort and select top-$k$ answers.  For each answer in the top-$k$ list, the last functional component builds the answer tree by finding the middle nodes from the path index and the textual information from the node text table.  At the end, the list of top-$k$ answer trees are shown to the user.  Below, we describe the search algorithms, indexing methods and data storage in detail.

\subsection{Search and Ranking Algorithms}

Since minimizing the node importance and combined objectives is NP-hard, we propose two greedy algorithms to solve them.  The first algorithm is based on converting the input graph into a new graph with importance/weights only on the edges.  The new graph $G'$ has the same set of nodes as the original graph $G$, but the node importance in $G$ is moved onto the edges in $G'$.  Therefore, the \textit{inverted index} and \textit{node text table} of $G$ and $G'$ are the same.  Assume that the edge weight between nodes $n_i$ and $n_j$ in $G$ is denoted by $w(n_i,n_j)$.  In $G'$, the edge weight between nodes $n_i$ and $n_j$ is defined as $w'(n_i,n_j) = \lambda.(imp'(n_i)+imp'(n_j)) + 2.(1-\lambda).w(n_i,n_j)$. Note that when $\lambda$ is set to 1, this algorithm only minimizes the node importance objective.  

For finding answers, we adapt the unique root semantic approach \cite{He2007}.  We assume each node $n_r$ in the graph could be a potential root for an answer tree.  To build a tree around $n_r$, for each given keyword $k_i$, we assign the closest node $n_{k_i}$ in $G'$ that contains $k_i$ to a tree rooted at $n_r$.  The tree with the lowest sum of the new edge weights is the best answer.  Since we follow the unique root semantic, and we run this polynomial operation on every node (i.e., potential root) of the graph once, the total run time of the algorithm is also polynomial.

To further improve efficiency, we introduce the second algorithm, which does not require to build separate indexes for different values of $\lambda$ and therefore, it saves space.  Here, we build a new graph $G''$ with the same set of nodes as $G$.  In $G''$, the edge weight between nodes $n_i$ and $n_j$ is defined as $w''(n_i,n_j) = imp'(n_i)+imp'(n_j)$.  This is similar to $G'$ when $\lambda=1$.  From the query keywords, we choose the one with the smallest set of content nodes.  Call this keyword $k$ and its set of content nodes $S_k$.  Then, for each node $n_i \in S_k$ and for the remaining query keywords, we select the content node closest to $n_i$ in $G''$.  A tree rooted at $n_i$ is formed and its $C(T)$ is calculated.  Then, for each query keyword other than $k$, we replace the content node of that keyword with the second closest content node to $n_i$ in $G''$.  If this improves the value of $C(T)$, this operation is applied permanently.  This process is repeated for a limited number of times or until the replace operation does not decrease $C(T)$ by more than a pre-defined threshold $\delta$.  Among the trees rooted at the content nodes of $S_k$, the one with the lowest $C(T)$ is chosen as the best answer.  Unlike our first algorithm, for different values of $\lambda$, we do not need to build a separate index for the input graph.  On the other hand, we show in the experiments that the results of the first algorithm is closer to the optimal results obtained by the exhaustive search algorithm.

By conducting a user study, we found that ranking answers using our node importance and combined objectives returns more meaningful results than previous methods.  In particular, in Figure~\ref{fig:graph}, both of the proposed algorithms return answer (b) as the best answer (with $\lambda$ set to 0.25, 0.5, 0.75 or 1).  Additionally, the experimental results show that our greedy algorithms produce answers whose quality is close to that of the exact algorithm obtained by exhaustive search.

\subsection{Indexing and Data Storage}

For an input graph $G$, the following four data structures are built to facilitate the search process.

\begin{enumerate}[topsep=0pt,itemsep=-1ex,partopsep=1ex,parsep=1ex]
\item \textit{Inverted Index}: For each keyword in the input graph $G$, the IDs of the nodes that contain the keywords are stored in an inverted index.  This index is stored in a DBMS.
  
\item \textit{2-Hop Cover Distance Index}: This index is used to find the shortest distance between two nodes.  We use the 2-hop cover index and store it in memory.
  
\item \textit{2-Hop Cover Path Index}: After finding the set of content nodes, we need to find the shortest path from the root to each of the content nodes.  Therefore, as part of the 2-hop index, we store extra information to retrieve the shortest path.  This index is also stored in main memory.
  
\item \textit{Node Text Table}: The textual information contained in each node is stored in this table.  For example, in the IMDb dataset, the titles of the movies and the names of the actors are stored in this table.  This table is maintained in a DBMS.
  
\end{enumerate}

\begin{figure*}[t]
\centering
\includegraphics[width=10cm]{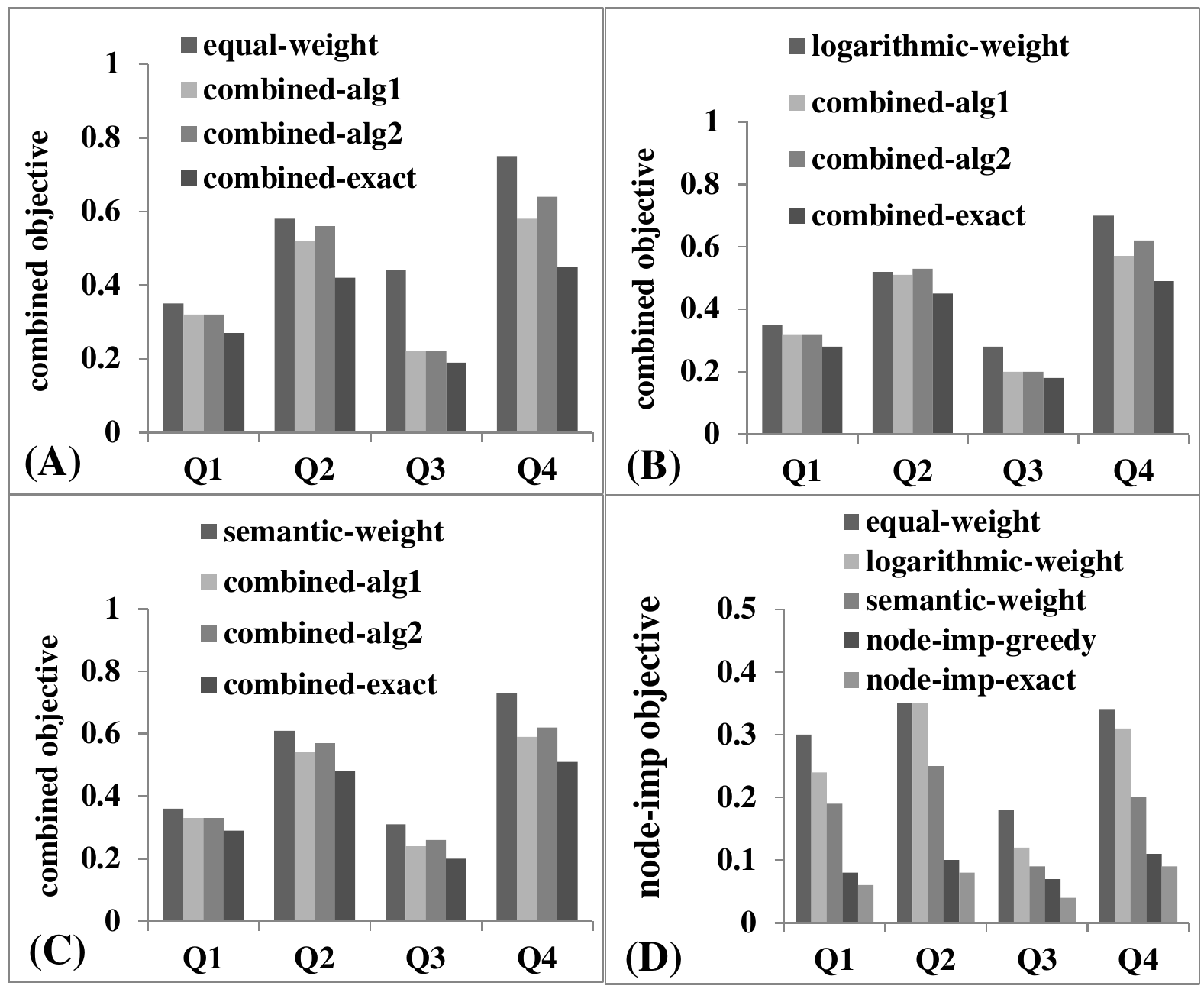}
\caption{Node importance and combined objectives of top-5 answers.} \label{fig:ResultsCombinedWeight}
\end{figure*}

\subsubsection{2-Hop Cover Index}
\label{Section:2-Hop-Cover-Index}

A distance query returns the distance between two nodes in graph $G$. For finding the answers to graph keyword search, we frequently need to find the distance between a potential root and a content node that contains at least one of the input keywords. One simple and inefficient way to do this is to compute the distance query on-the-fly. This can be done using a breadth first search (BFS) in unweighted graphs or Dijkstra's algorithm in weighted graphs. However, finding the distance between one pair of nodes might take more than one second for large graphs (number of nodes/edges $> 1M$) \cite{Akiba2013}. On the other hand, we have to find the distances between a number of pairs of nodes to find answers for each query. Since we need to have low latency due to real-time interactions between users and the application, this method is too slow to use as a building block of our application. The other simple and inefficient way is to compute distances between every pair of nodes in advance and store them in an index. In a graph with $N$ nodes, this method needs $O(N^2)$ space to store the index. Thus, for big graphs, we quickly run out of memory. Furthermore, computing this index takes a long time on large graphs. Therefore, even though we can answer distance queries instantly, this method is not feasible due to large space and long pre-processing time. Thus, from the performance point of view (both index size, pre-processing time and query time), we need to use a practical index that lies between these two extreme approaches.

In this demonstration, we adapt a method that is based on the concept of \textit{distance labeling} or \textit{2-hop cover} \cite{Cohen2002}. The general idea of 2-hop cover is as follows. For each node $n_i$ in graph $G$, a label $L(n_i)$ is computed and stored. $L(n_i)$ is a set of pairs $(n_j, d(n_j,n_i))$ where $n_j$ is a node in $G$ and $d(n_j,n_i)$ is the shortest distance between $n_j$ and $n_i$. The set of labels $\{L(n_i)\} | n_i \in N_G$ is called the 2-hop cover index ($N_G$ is the set of all nodes of $G$). To answer a distance query between nodes $n_s$ and $n_t$, $Dist(n_s, n_t, L)$ is computed as follows.

\begin{multline*}
  Dist(n_s, n_t, L) = \min \{d(n_l,n_s)+d(n_l,n_t) \ | \\ 
  (n_l,d(n_l,n_s))\in L(n_s),(n_l,d(n_l,n_t))\in L(n_t)\} 
\end{multline*}

$Dist(n_s, n_t, L)$ returns $\infty$ if $n_s$ and $n_t$ are not connected in $G$. $L$ is called a \textit{distance-aware 2-hop cover} of $G$ if $Dist(n_s, n_t, L)$ returns the shortest path between any pair of nodes $n_s$ and $n_t$ in $G$. For each node $n_i$, the label $L(n_i)$ is stored as a sorted list based on the order of the nodes. Then, $Dist(n_s, n_t, L)$ can be computed in $O(|L(n_s)| + |L(n_t)|)$ time using a a merge join algorithm \cite{Cohen2002}.

\begin{figure}
\centering
\includegraphics[width=8.5cm]{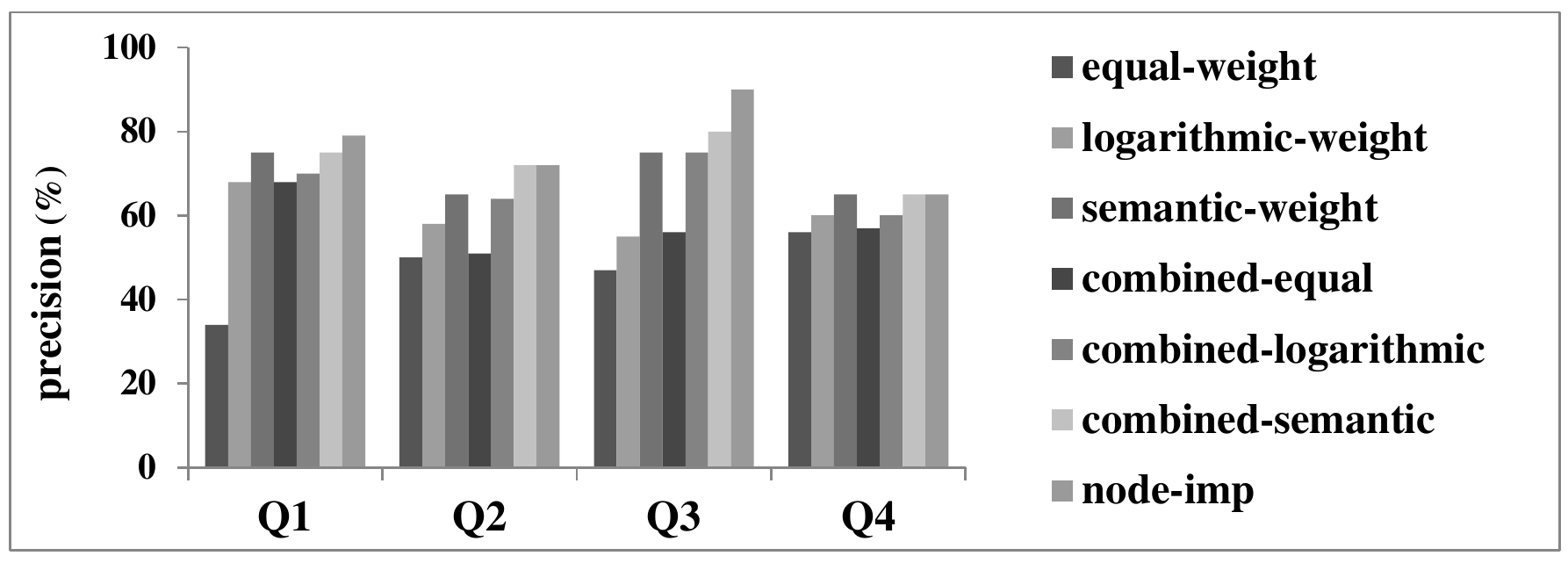}
\caption{Precision of top-5 answers.  $\lambda$ is set to 0.5.} \label{fig:ResultsPrecision}
\end{figure}

For producing the 2-hop cover list, we use a modification of the algorithm in Akiba et. al., \cite{Akiba2013}. The idea is as follows. First, the nodes are sorted based on their degree. This helps to decrease the size of the index and decrease the distance query time. Let $N = \{n_1, n_2, \dots , n_N\}$ be the set of nodes sorted by their degree (i.e., $n_1$ is the node with highset degree). The algorithm starts with an empty index $L_0$ in which $L_0(n_i) = \emptyset \ \forall \ n_i \in N_G)$. Then, we conduct pruned Dijkstra's algorithm\footnote{As we explain later, in this work, most of our input graphs are weighted graphs. Therefore, in order to build a 2-hop index, we need to use Dijkstra's algorithm. When using an unweighted graph, BFS can be used instead of Dijkstra to decrease the pre-processing time of builing the index.} from nodes in the order of $n_1$, $n_2$, $\dots$, $n_k$. After the $i$-th Dijkstra from a node $n_i$, the distances from $n_i$ to labels of the reached nodes are added to the $L_i$ ($L_i(n_l) = L_{i-1}(n_l) \cup \{(n_i, d(n_i,n_l))\}$). At the end of the $N$-th operation, $L_N$ is the final index. To further improve the performance of this index, we only keep the distances up to a threshold $D_{max}$ (by default, we set this threshold to ten times the average value of all the edge weights in the input graph). In other words, during the index creation, if $d(n_i,n_j) > D_{max}$, we assume $n_i$ and $n_j$ are disconnected and we stop further expanding the Dijkstra algorithm on $n_i$. Since we are only interested to find answers in which the content nodes are reasonably close to each other (and close to the root node), we do not need to calculate the distance between nodes which are far away from each other. $D_{max}$ determines the maximum distance between any content node and the root of the tree and it depends on the application. 
%

In order to find the shortest path between any pair of nodes, in the 2-hop cover index, we store a triple instead of a pair as a label. In other words, $L(n_i)$ is the set $(n_l, d(n_l,n_i), p(n_l,n_i))$, where $p(n_l,n_i) \in N_G$ is the parent of $n_l$ in the pruned Dijkstra search tree with root $n_l$. The shortest path between $n_l$ and $n_i$ can be restored by ascending the tree from $n_i$ to the parents.

\begin{figure}
\centering
\includegraphics[width=8.5cm]{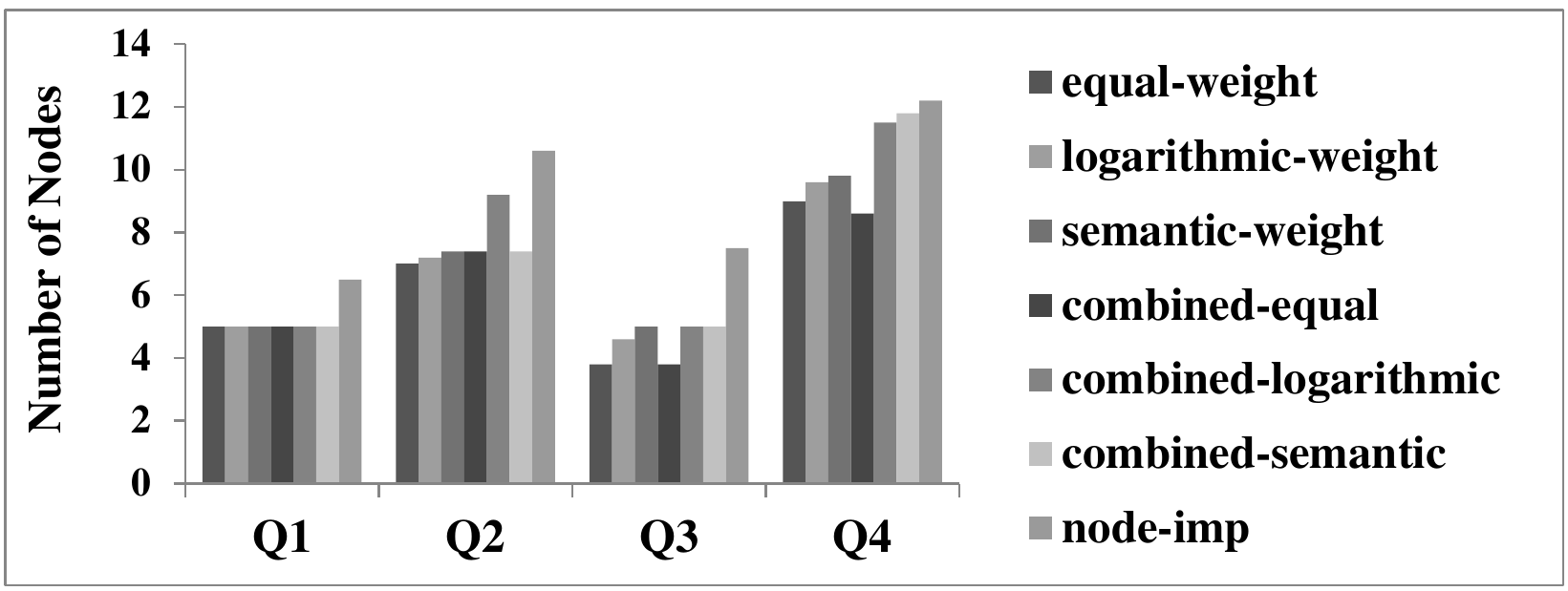}
\caption{Number of Nodes of top-5 answers.  $\lambda$ is set to 0.5.} \label{fig:ResultsNumberOfNodes}
\end{figure}

Our experiments show that the 2-hop cover index controlled by $D_{max}$ parameter scales well for large graph databases (Section 4).  The index can be stored in main memory or in an SSD.  This optimization and applying the second algorithm that is $\lambda$ independent turns to be an effective approach for building a scalable system without sacrificing the quality of answers.

\subsection{Edge and Node Weights}
\label{Section:Edge_Weight_Calculation}

In this work, we include three ways of assigning weights to edges: 1) equal weights, 2) logarithmic weights, in which the weight of an edge between two nodes $n_i$ and $n_j$ is $(\log_2(1 + deg_{n_i}) + \log_2(1 + deg_{n_j}))/2$, where $deg_{n}$ is the degree of node $n$, and 3), application-dependent semantic weights, in which the weight of an edge is related to the semantic relation between the two end nodes of the edge.  In the IMDb dataset, we use the weight of the associated foreign keys between two tuples in the relational database as the semantic weights between the associated nodes in the input graph, computed using the algorithm from \cite{Kargar2015a}. Node weights (node importance values) are application-dependent.  As we mentioned earlier, for IMDb, we calculate node importance using the algorithm from \cite{Kargar2015a}.

\section{Experiments}

We now test three algorithms on the IMDb dataset\footnote{\textit{http://imdbpy.sourceforge.net}}: 1) the algorithm from \cite{He2007} which only considers edge weights, 2) a version of our first greedy algorithm with $\lambda=1$, labeled node-imp, which only optimizes the node importance, and 3), our greedy algorithm, labeled combined, with $\lambda \in (0,1)$.  As we will show next, the results of the first algorithm that optimizes the combined objective is slightly closer to the exact method than the second one.  Therefore, if not explicitly stated, the combined algorithm refers to the first proposed greedy algorithm that optimizes the combined objective.  As we mentioned in Section \ref{Section:Edge_Weight_Calculation}, we use three ways to assign weights to the edges. Therefore, the results of first algorithm that minimizes the edges weights are presented for each of these ways.  They are called equal-weight, logarithmic-weight and semantic-weight in this section.  The IMDb graph has 1 million nodes and 3 million edges.  All the algorithms are implemented in Java.  The experiments are conducted on an Intel(R) Core(TM) i7 2.80 GHz computer with 8.0 GB of RAM.

As mentioned before, we calculate node importance using the algorithm from \cite{Kargar2015a}. Node importance values are normalized between 0 and 1. The average of the node importance values is 0.1. Since all the edge weights have the same value in IMDb, and in order to have a fair distribution of values for both node importance and edge weights, we set all the edge weights to 0.1.  We use the following queries over movie and actor names: Q1= ["Morgan Freeman" "Tim Robbins"], Q2=["Morgan Freeman" "Tim Robbins" "Keanu Reeves"], Q3=["The Matrix" "The Shawshank Redemption"] and Q4=["The Matrix" "The Shawshank Redemption" "Morgan Freeman" "Keanu Reeves"].  The IR scores of all of the content nodes are the same.  Thus, the results of the IR based method is similar to the results of the equal-weight method.  Therefore, to avoid repetition, the results of the IR based method are omitted here.

\begin{figure*}[t]
\centering
\includegraphics[width=18cm]{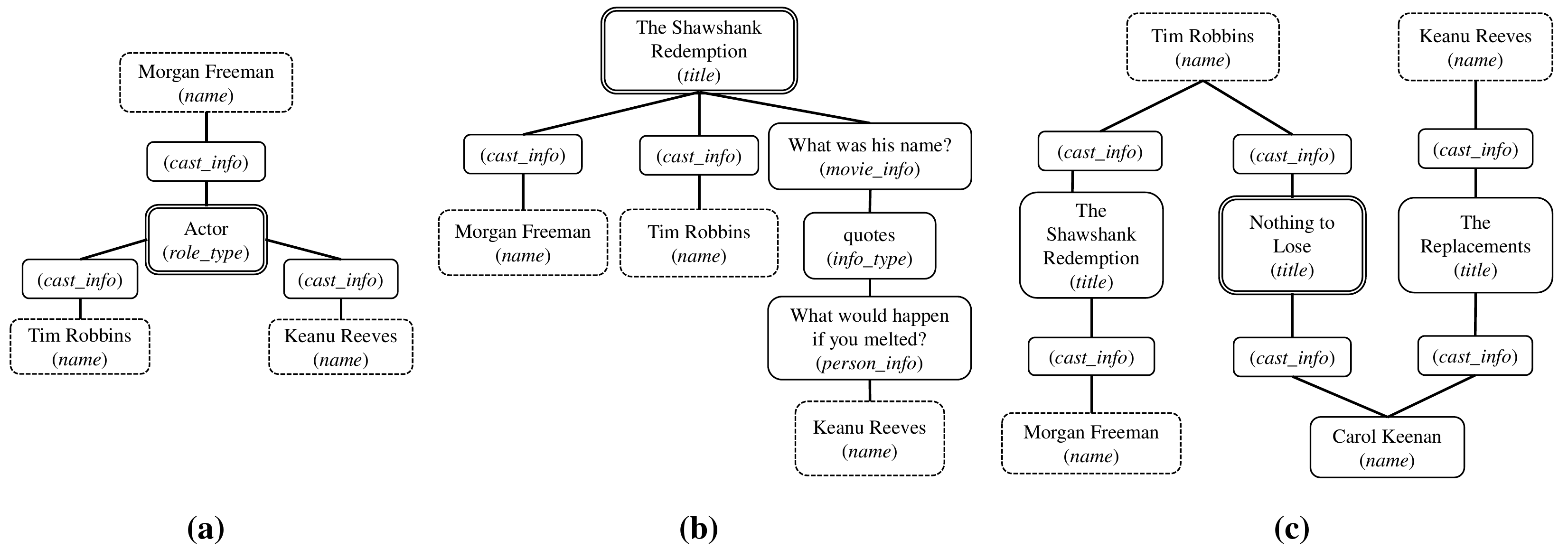}
\caption{Three top answers returned by different ranking strategies.  Content nodes are shown in dotted boxes.  Root nodes are shown in double-lined boxes.} \label{fig:QualitativeEvaluation}
\end{figure*}

The average value of the node importance and combined objectives for edge-based methods, combined (both algorithms and $\lambda$=$0.5$) and node-imp algorithms are shown in Figure \ref{fig:ResultsCombinedWeight}. We also present the results of the exact algorithm, obtained using exhaustive search. In practice it is too expensive to use it as users expect to receive answers in real time. Note that we have two gold standards, one that minimizes the node importance and one that minimizes the combined objective with $\lambda$ set to 0.5 with different strategies for assigning edge weights. The results suggest that our greedy algorithms produce answers whose quality is close to that of the exact algorithm. In addition, the results of the first proposed greedy algorithm for minimizing the combined objective is slightly closer to the exact results than the second greedy algorithm.  We also note that all three algorithms have similar runtime since they all use the same indexing method.

We compare the edge-based, node-imp and combined algorithms in terms of how relevant their answers are to the query. A common metric of relevance used in information retrieval is top-$k$ precision. It is the percentage of the returned answers in the top-$k$ answers that are relevant to the query. To evaluate the top-$k$ precision of the algorithms, we designed and performed a user study. We asked 8 users interested in movies (each with a different educational background) to judge the relevance of the top-5 answers using a score between zero and one.  Figure \ref{fig:ResultsPrecision} shows the top-5 precision of each tested algorithm on each query.  For the combined algorithm, we used different ways to assign edge weights and the value of $\lambda$ is set to 0.5.  Both of our methods, combined and node-imp, obtain better precision than equal-weight for all queries.  The logarithmic-weight and semantic-weight produce better results than equal-weight.  However, node-imp outperforms logarithmic-weight and semantic-weight and combined for all queries.  
%
%

The number of nodes for different algorithms is presented in Figure \ref{fig:ResultsNumberOfNodes}. The edge-equal has the smallest number of nodes. This is because all of the edges have the same weight in IMDb and final answers are presented in tree form. In trees, the number of edges are equal to the number of nodes minus one. Therefore, by minimizing the sum of the edge weights, the number of edges/nodes are implicitly minimized too. The node-imp has the highest number of nodes. The reason is that different nodes have different importance values. Therefore, we may have answers with many nodes and edges but small value of the sum of the node importance. By incorporating the edge weights into the node importance, we devalue the node importance values. Therefore, the number of nodes in the combined algorithm falls between equal-weight and node-imp.

We also design the following experiment to show that our system is able to save a significant amount of time to query databases even for expert users well-versed in SPARQL.  We measure the time to perform search both manually by four experts and automatically with our keyword search system.  To simulate real word environment with the manual search, experts were allowed to access the tools that they routinely use when they query graph databases.  When we measure the time for automatic search with our tool, we include the time for specifying the parameters, as well as preforming the actual search by our system and scrolling through the ranked answers to find the meaningful ones.  Whereas, when we measure the time for manual search, we include the time both for writing the SPARQL query and executing the prescribed query. We observed that our tool reduces around 10 times the time to search for desirable answers.  The automatic search was on average performed in 8 seconds, versus the manual search in around 90 seconds.  As expected, the manual search is prone to human errors, such as mistakes in syntax or misunderstanding the schema that introduces delays.

\begin{table}
\centering
\caption{Performance comparison of the 2-hop cover index for the IMDb graph with equal weights and different values of $D_{max}$ parameter. All of the edges have the weight of 1.0.}
\begin{tabular}{|l|l|l|} \hline
$D_{max}$ \ \ \ \ \ \ \ \ \ \ \ \ \ \ \ & Index Size \ \ \ \ \ \ \ \ \ \ & Query Time \ \ \ \ \ \ \ \ \ \  \\ \hline
3 & 399 MB& 0.8 $\mu s$  \\ \hline
5 & 429 MB& 0.9 $\mu s$  \\ \hline
7 & 435 MB& 0.9 $\mu s$  \\ \hline
9 & 450 MB& 0.9 $\mu s$  \\ \hline
$\infty$ & 450 MB& 0.9 $\mu s$  \\ 
\hline\end{tabular}
\label{tbl:indexEqualWeight}
\end{table}

\subsection{Qualitative Evaluation}

We compare the edge-weight, combined and node-imp algorithms via an example.  Using the IMDb dataset and assuming that edges have equal weights, the top answer returned by edge-weight ranking and IR-based ranking for the query ["Morgan Freeman" "Tim Robbins" "Keanu Reeves"] is shown in Figure~\ref{fig:QualitativeEvaluation} (a).  The top answer returned by our combined method (first greedy algorithm and $\lambda = 0.5$) for the same query is shown in Figure~\ref{fig:QualitativeEvaluation} (b).  Finally, the top answer returned by our node-importance and combined methods (first greedy algorithm and $\lambda = 0.75$) is shown in Figure~\ref{fig:QualitativeEvaluation} (c).  Answer (a) suggests that all three of these actors are connected merely because they are actors.  Answer (b) shows that Morgan Freeman co-starred in The Shawshank Redemption with Tim Robbins.  Furthermore, Keanu Reeves and The Shawshank Redemption both have an \emph{info\_type} of ``quote'' (a memorable quote from the movie) and that is how all the content nodes are connected together.  Answer (b) is more interesting than answer (a).  However, having the same \emph{info\_type} for an actor and a movie might not be interesting.  The last answer (answer (c)) reveals that Morgan Freeman starred in the same movie with Tim Robbins (The Shawshank Redemption).  Furthermore, it shows that Tim Robbins and Carol Keenan played in the same movie (Nothing to Lose).  Keanu Reeves is connected to the other two actors by playing in the same movie with Carol Keenan (The Replacements).  We argue that answer (c) is more interesting than answer (a) and (b) and reveals a deeper connection between the content nodes that may not have been known to a user who is unfamiliar with the dataset. 

\begin{table}
\centering
\caption{Performance comparison of the 2-hop cover index for the IMDb graph with combined method and different values of $D_{max}$ parameter. The edge weights are equal and the $\lambda$ parameter is set to 0.5. The maximum edge weight is 1.51 and the minimum edge weight is 1.01.}
\begin{tabular}{|l|l|l|} \hline
$D_{max}$ \ \ \ \ \ \ \ \ \ \ \ \ \ \ \ & Index Size \ \ \ \ \ \ \ \ \ \ & Query Time \ \ \ \ \ \ \ \ \ \  \\ \hline
3 & 522 MB& 0.9 $\mu s$  \\ \hline
5 & 750 MB& 1.1 $\mu s$  \\ \hline
7 & 1.75 GB& 1.6 $\mu s$  \\ \hline
9 & 1.78 GB& 1.6 $\mu s$  \\ \hline
$\infty$ & 1.78 GB& 1.6 $\mu s$  \\ 
\hline\end{tabular}
\label{tbl:indexCombined}
\end{table}

\subsection{Scalability of the Index}

In this section, we present the scalability studies regarding our 2-hop cover index.  As mentioned before, the IMDb graph has 1 million nodes and 3 million edges.  Tables \ref{tbl:indexEqualWeight}, \ref{tbl:indexCombined} and \ref{tbl:indexNodeImp} show the performance of the index for the IMDb graph with equal-weight, combined and node-imp methods respectively.  Query time denotes the average time for 1,000,000 random distance queries.  The $D_{max}$ value for node-imp is set around 0.1 because the distribution of the edge weights are different than the other two methods in node-imp.  The run rime suggests that the index answers distance queries almost instantly.  The results suggest that the $D_{max}$ does not have a significant effect on a graph when using equal-weight and combined methods.  However, it has a significant effect on a graph in which the node-imp method is used.  Table \ref{tbl:indexNodeImp} shows that the size of the index increases drastically when $D_{max}$ is changed from 0.07 to 0.1.  Figures \ref{fig:ScalabilityEqualEdge}, \ref{fig:ScalabilityCombined} and \ref{fig:ScalabilityNodeImp} show the precision of the answers with different values of $D_{max}$.  As the results suggest, using smaller values of $D_{max}$ does not sacrifice the precision.  For example, as we show in Figure \ref{fig:ScalabilityNodeImp}, the precision of the answers with $D_{max}$ equal to 0.07 is equal to the one with $D_{max}$ is set to 0.1.  Therefore, we can save space using smaller values of $D_{max}$ while not losing the precision of the answers.

\begin{table}
\centering
\caption{Performance comparison of the 2-hop cover index for the IMDb graph with node importance transferred to edges and different values of $D_{max}$ parameter. The maximum edge weight is 1.02 and the minimum edge weight is 0.01.}
\begin{tabular}{|l|l|l|} \hline
$D_{max}$ \ \ \ \ \ \ \ \ \ \ \ \ \ \ \ & Index Size \ \ \ \ \ \ \ \ \ \ & Query Time \ \ \ \ \ \ \ \ \ \  \\ \hline
0.03 & 440 MB& 0.9 $\mu s$  \\ \hline
0.05 & 980 MB& 1.1 $\mu s$  \\ \hline
0.07 & 1.35 GB& 1.5 $\mu s$  \\ \hline
0.10 & 7.82 GB& 2.1 $\mu s$  \\ \hline
$\infty$ & NA & NA  \\ 
\hline\end{tabular}
\label{tbl:indexNodeImp}
\end{table}

\begin{figure}
\centering
\includegraphics[width=7.0cm]{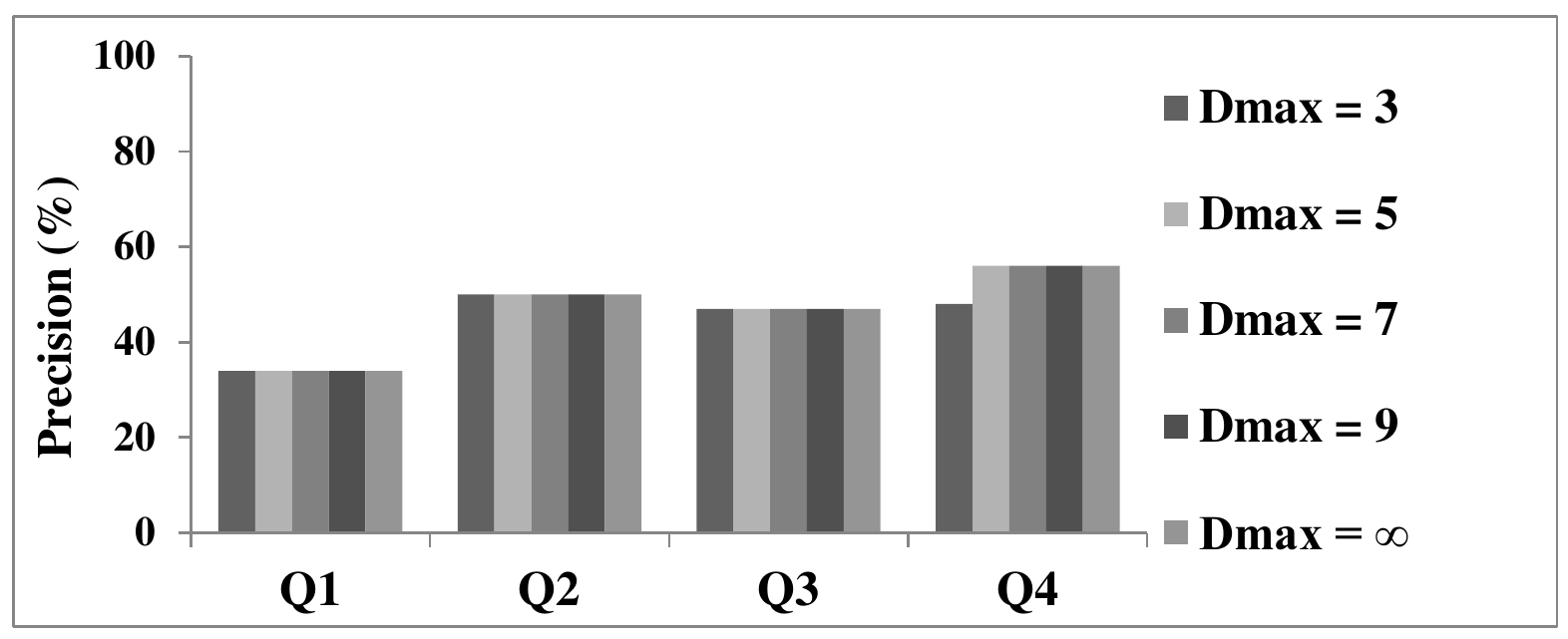}
\caption{Precision of top-5 answers using equal weight method for different values of $D_{max}$.} \label{fig:ScalabilityEqualEdge}
\end{figure}

\begin{figure}
\centering
\includegraphics[width=7.0cm]{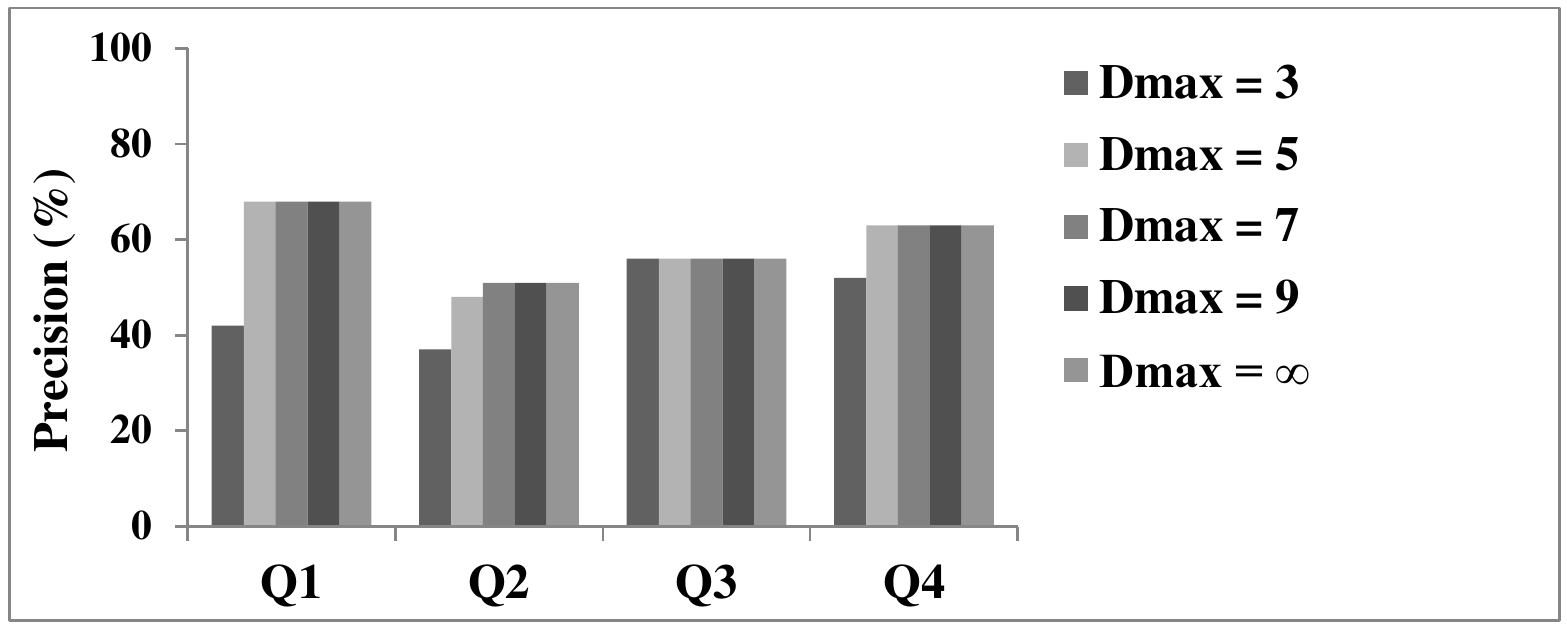}
\caption{Precision of top-5 answers using combined method for different values of $D_{max}$. The edge weights are equal and the $\lambda$ parameter is set to 0.5.} \label{fig:ScalabilityCombined}
\end{figure}

\begin{figure}
\centering
\includegraphics[width=7.0cm]{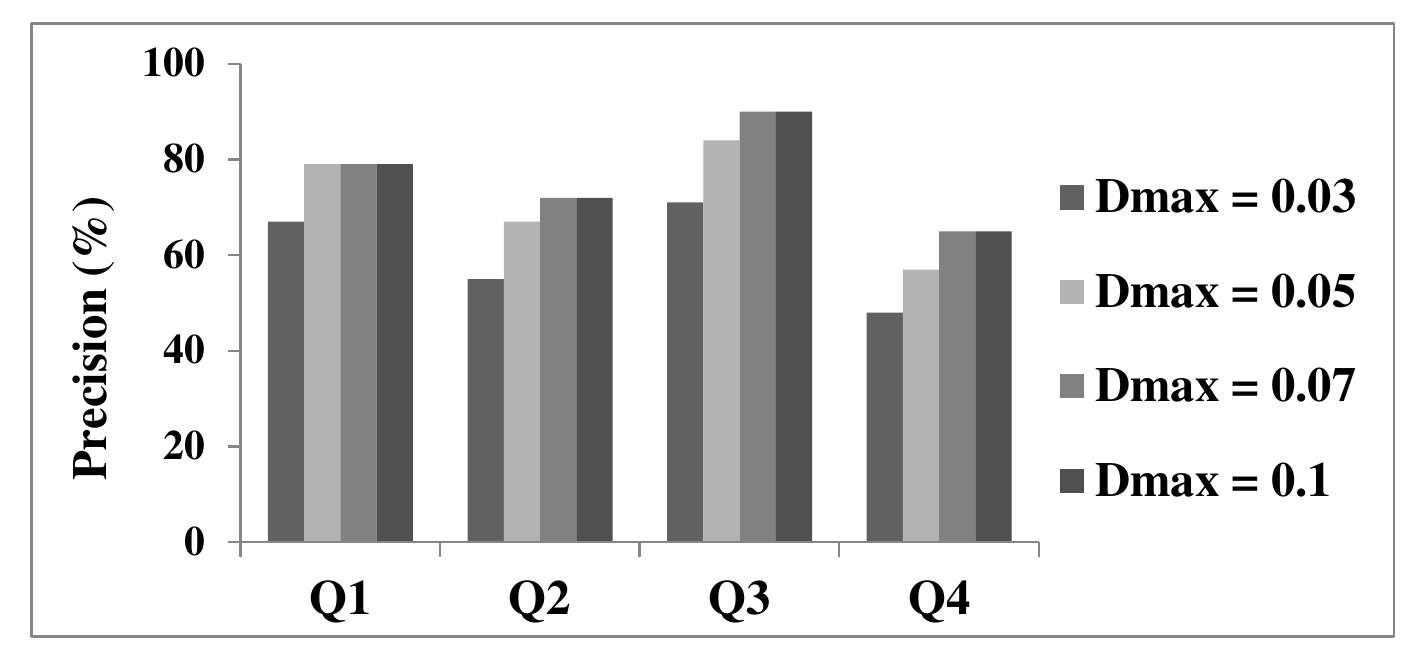}
\caption{Precision of top-5 answers using node-imp method for different values of $D_{max}$.} \label{fig:ScalabilityNodeImp}
\end{figure}

\section{Conclusions}

In this work, we introduce the problem of finding effective answers for keyword search over graphs in the existence of node importance. We define the problem of minimizing the node importance and proved that it is NP-hard. We also define a combined objective function that combines the values of the node importance and edge weights. For minimizing the node importance and combined objective functions, we experimentally verify that the result of our proposed greedy algorithms are close to the gold standards.

While the experiments verify the efficiency of our proposed greedy algorithms, in the future, we will work on establishing theoretical bounds for approximation algorithms. We plan to work on designing alternative strategies (such as an iterative approach) and more efficient indexes for finding answers. We will also apply our approach to different datasets with edge semantics. Given both node and edge semantics, we will examine which method will produce more meaningful results.

\bibliographystyle{abbrv}
\bibliography{biblioGraph}

\end{document}